\documentclass[a4paper,12pt,twoside]{amsart}
\usepackage[all]{xy}
\usepackage[utf8]{inputenc}
\usepackage{braket}
\usepackage{caption}
\usepackage{subcaption}

\usepackage{amsmath,amssymb,latexsym}
\usepackage{graphicx}
\usepackage{amsthm,url,color}
\usepackage[export]{adjustbox}
\newtheorem{theorem}{Theorem}[section]

\newtheorem{corollary}[theorem]{Corollary}
\newtheorem{proposition}[theorem]{Proposition}

\theoremstyle{definition}

\newtheorem{example}[theorem]{Example}

\theoremstyle{remark}
\newtheorem{remark}[theorem]{Remark}

\numberwithin{equation}{section}

\newcommand{\C}{\mbox{$\Bbb C$}}

\renewcommand{\phi}{\varphi}

\newcommand{\R}{\mathbb R}
\newcommand{\G}{\mathbb G}
\newcommand{\Ss}{\mathbb S}

\begin{document}

%	\begin{titlepage}
		
		\title{Quantum computing based on complex Clifford algebras}
		\thanks{The authors was supported by the grant no. FSI-S-20-6187}

		\maketitle
		
		\begin{center}
			
			\author{\bf Jaroslav Hrdina, Ale\v s N\' avrat, Petr Va\v s\'ik
			}
			\vspace{6pt}
			
			\small
			
			Institute of Mathematics,
			Faculty of Mechanical Engineering, \\ Brno University of Technology,
			%\\
			%Technick\' a 2896/2, 616 69 Brno, 
		Czech Republic
			\\	\texttt{hrdina@fme.vutbr.cz, navrat.a@fme.vutbr.cz, vasik@fme.vutbr.cz}

		\end{center} 
		
		\thispagestyle{empty}

		\vspace{3pt}

	%\end{titlepage}
	
	\begin{abstract}
	 We propose to represent both $n$--qubits and quantum gates acting on them as elements in the complex Clifford algebra defined on a complex vector space of dimension $2n.$  
	 In this framework, the Dirac formalism can be realized in straightforward way. 
	 	 We demonstrate its functionality by performing quantum computations with several well known examples of quantum gates. We also compare our approach with representations that use real geometric algebras. 
	\end{abstract}
		
\section{Introduction}
Real geometric (Clifford) algebras (GA) may be understood as a generalisation of well known quaternions which is an alternative for matrix description of orthogonal transformations. Real geometric algebras have a wide range of applications in robotics, \cite{joan,hdita}, image processing \cite{stev}, numerical methods \cite{hfit}, etc. Among the main advantage of this approach we count the calculation speed, straightforward and geometrically oriented implementation and effective parallelisation, \cite{dita, dita2}. We stress that all these implementations are using Clifford's geometric algebra to represent specific orthogonal transformations. 

Recently, an increasing number of papers adopt the apparat of real GA in the description, elaboration and analysis of Quantum Computing (QC) algorithms, \cite{cm, qba1}. The basic idea for this lies in identification of a state qubit with a Bloch sphere together with the identification of qubit gates with the rotations of the sphere. Namely, it is well known that a normalized qubit can be written in terms of basis vectors $\ket{0},\ket{1}$ as 
\begin{align}\label{norm_qubit}
\ket{\psi}=\cos(\theta/2)\ket{0}+(\cos\phi +i \sin\phi)\sin (\theta/2)\ket{1},
\end{align}
where $0\leq \theta \leq \pi$ and $0\leq \phi < 2\pi$ and that parameters $\theta,\phi$ can be interpreted as spherical coordinates of a point on the unit sphere  in $\R^3,$ see figure \ref{Bloch}. 
\begin{figure}[h] 
	\centering
	\includegraphics[width=5cm]{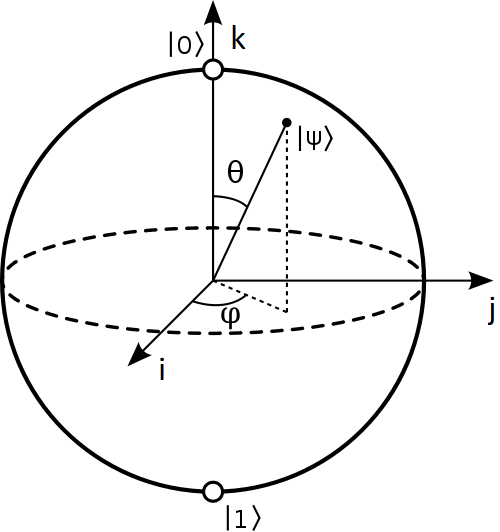}
	\caption{Bloch sphere for a qubit represented in quaternions}
	\label{Bloch}
\end{figure}
The spirit of GA description is to represent a point on the sphere by the rotation that takes a fixed initial point to that point. Under the usual choices, see figure \ref{Bloch},  the initial point is the north pole $(0,0,1)$ and  the qubit \eqref{norm_qubit} corresponds to the counterclockwise rotation by $\theta$ with respect to $y$ axis composed with the clockwise rotation by $\phi$ with respect to $z$ axis. In particular, the basis states $\ket{0},\ket{1}$ are represented by the identity and rotation $\theta=\pi$ respectively.
In the quaternionic description we have $(0,0,1)={\bf k}$ and the rotations are represented by elements $\exp(-1/2\theta  {\bf j})=\cos(\theta/2)-{\bf j}\sin(\theta/2)$ and $\exp(1/2\phi {\bf k})=\cos(\phi/2)+{\bf k}\sin(\phi/2)$ respectively. Indeed, the Bloch sphere representation of qubit \eqref{norm_qubit} is given by
\begin{align*}
e^{\tfrac12\phi {\bf k}}e^{-\tfrac12\theta {\bf j}} {\bf k} e^{\tfrac12\theta {\bf j}} e^{-\tfrac12\phi {\bf k}}=\cos\phi\sin\theta {\bf i}+\sin\phi\sin\theta {\bf j}+\cos\theta {\bf k}
\end{align*}
In this sense the qubit state \eqref{norm_qubit} is represented by quaternion $\exp(1/2\phi {\bf k})\exp(-1/2\theta  {\bf j})$.
In particular, the basis states $\ket{0},\ket{1}$ are represented by quaternions $\ket{0}=1$ and $\ket{1}=\exp(-1/2\pi  {\bf j})={-\bf j}$ respectively which in GA language is equation \eqref{basisH}, see section \ref{H_qubit} for more details. 
In principle, this representation of qubits is based on exceptional isomorphisms of low-dimensional Lie groups $SU(2)$, $Spin(3)$ and $Sp(1)$. For higher dimensions, no similar natural identification of unitary and spin groups exists. Therefore, to realise the states of multi-qubit, it is necessary to use so-called correlators which increases algebraic abstraction and lacks the original elegance. Indeed, such approach is no competition to elegant Dirac formalism.

In our paper, we present an alternative to the qubit representation in the form of complex Clifford algebras which are indeed substantially more appropriate for QC representation by means of GA because they respect the complex nature of quantum theory. In our approach, Dirac formalism may be translated to complex GA of an appropriate dimension in rather straightforward way and, in spite of abstract symbolic Dirac formalism, all expressions are represented in a particular algebra and thus may be manipulated and implemented as algebra elements directly, without any need for matrix representation.  
In the sequel, we briefly recall the definition of real GA and provide a more detailed introduction to complex GA. In Section \ref{QCinC}, we show the representation of qubits and multi-qubits in complex GA, their transformations (gates).  We provide an explicit forms of elementary 1-gates and 2-gates. We also discuss the case of multi-gates, ie. gates obtained by a tensor product. In Section \ref{realGA}, we describe a qubit by means of real geometric algebra. More precisely, we compare a description known from literature, ie. the one based on the isomorphism of unitary group $SU(2)$ and spin group $Spin(3)$, with the real description following from our complex GA approach and an isomorphism 
$\C_2\rightarrow \G_3$ of complex and real algebra.

\section{Complex Clifford algebras}	\label{CCA}
A Clifford algebra is a normed associative algebra that generalizes the complex numbers and the quaternions. Its elements may be split into Grassmann blades and the ones with grade one can be identified with the usual vectors. The geometric product of two vectors is a combination of the commutative inner product and the anti-commutative
outer product. The scalars may be real or complex. In  section \ref{QCinC} we show that the complex Clifford algebra constructed over a quadratic space of even dimension can be efficiently used to represent quantum computing but we start with the real case.

\subsection{Real Clifford algebras} 
The construction of the universal real Clifford algebra is well-known, for details see e.g. \cite{dl,lounesto}. We give only a brief description here. Let the real vector space $\R^m$ be endowed with a non-degenerate symmetric bilinear form $B$ of signature $(p,q)$, and let $(e_1,\dots, e_m)$ be an associated orthonormal basis, i.e.
\begin{align*}
B(e_i,e_j)=\begin{cases}
1 &\text{ if } i=j=1,\dots,p\\
-1 &\text{ if } i=j=p+1,\dots,m\\
0  &\text{ if } i\neq j
\end{cases}
\; \text{ where } 1\leq i,j \leq m=p+q.
\end{align*}
Let us recall that the Grassmann algebra is an associative algebra with the anti-symmetric outer product defined by the rule 
\begin{align*}
e_i \wedge e_j + e_j \wedge e_i =0 \; \text{ for } 1\leq i,j \leq m.
\end{align*}
The Grassmann blade of grade $r$ is  $e_A=e_{i_1} \wedge \cdots \wedge e_{i_r}$, where the multi-index $A$ is a set of indices ordered in the natural way $1\leq i_1\leq \cdots \leq i_r\leq m,$ and we put $e_\emptyset=1.$ Blades of orders $0\leq r \leq m$ form the basis of the graded Grassmann algebra $\Lambda(\mathbb{R}^m)$. Next, we introduce the inner product
\begin{align*}
e_i \cdot e_j= B(e_i,e_j), \;\;  1\leq i,j \leq m,
\end{align*}
leading to the so-called {\it geometric product} in the Clifford algebra
\begin{align*}
e_i e_j= e_i\cdot e_j + e_i \wedge e_j, \;\;  1\leq i,j \leq m,
\end{align*}
The respective definitions of the inner, the outer and the geometric product are then extended to blades of the grade $r$ as follows. For the inner product we put
\begin{align*}
e_j \cdot e_A = e_j \cdot (e_{i_1} \wedge \cdots \wedge e_{i_r})= \sum_{k=1}^r (-1)^k B(j,i_k) e_{A\setminus \{i_k\}},
\end{align*}
where $e_{A\setminus \{i_k\}}$ is the blade of grade $r-1$ created by deleting $e_{i_k}$ from $e_A$. This product is also called the left contraction in literature. For the outer product we have
\begin{align*}
e_j \wedge e_A = \begin{cases}
e_j \wedge e_{i_1} \wedge \cdots \wedge e_{i_r} & \text{ if } j\notin A \\
0  & \text{ if } j\in A
\end{cases}
\end{align*}
and for the geometric product we define
\begin{align*}
e_je_A=e_j\cdot e_A + e_j \wedge e_A.
\end{align*} 
Finally, these definitions are linearly extended to the whole of the vector space $\Lambda(\mathbb{R}^m)$. Thus we get an associative algebra over this vector space, the so-called real Clifford algebra, denoted by $\mathbb{G}_{p,q}=\operatorname{Cl}(p,q,\mathbb{R})$. Note that this algebra is naturally graded; the grade zero and grade one elements are identified with $\mathbb{R}$ and $\mathbb{R}^m$ respectively. The projection operator $\G_{p,q}\to\Lambda^r(\R^m)$ will be denoted by $[\;]_r$.

This grading define a $\mathbb{Z}_2$-grading of the Clifford algebra according to the parity of grades. Namely, the linear map $v\to -v$ on $\mathbb{R}^m$ extends to an automorphism $\alpha$ called the grade involution and decomposes $\mathbb{G}_{p,q}$ into positive and negative eigenspaces. The former is called the even subalgebra $\mathbb{G}^0_{p,q}$ and the latter is called the odd part $\mathbb{G}^1_{p,q}$. In addition to $\alpha$, there are two important antiautomorphisms of real Clifford algebras. The first one is $\tilde{x}$ called the reverse or transpose operation and it is defined by extension of identity on $\mathbb{R}^m$ and by the antiautomorphism property  $\widetilde{xy}=\tilde{y}\tilde{x}.$ The second antiautomorphism is called the Clifford conjugation $\bar{x}$ and  the operation is defined by composing $\alpha$ and the reverse  
\begin{align*}
\bar{x}=\alpha(\tilde{x})=\widetilde{\alpha(x)}.
\end{align*}

\subsection{The complexification} \label{ClC}
When allowing for complex coefficients, the same generators $e_A$ produce by the same formulas the complex Clifford algebra which we denote by $\mathbb{C}_m=\operatorname{Cl}(m,\mathbb{C})$. Clearly, in the complex case no signature is involved, since each basis vector $e_j$ may be multiplied by the imaginary unit $i$ to change the sign of its square. Hence we may assume we start with the real Clifford algebra $\mathbb{G}_m$ with the inner product $e_ie_j=\delta_{ij}$, $1\leq i,j \leq m$, and we construct the complex Clifford algebra as its complexification $\C_m:=\G_m\oplus i \G_m,$ i.e. any element  $\varphi\in\C_m$ can be written as $\varphi=x+iy$, where $x,y\in\G_m.$  The complex Clifford algebras for small $m$ are well known; $\C_0$ are complex numbers itself, $\C_1$ is the algebra of bicomplex numbers and 
$\C_2$ is the algebra of biquaternions. More details on complex Clifford algebras one can find in the papers \cite{bss2,bss,bu,fs}.

The construction via the complexification of $\G_m$ leads to the definition an important antiautomorphism of $\C_m,$ so-called Hermitian conjugation
\begin{align} \label{HC}
\varphi^\dagger=(x+iy)^\dagger=\bar{x}-i\bar{y},
\end{align}
where the bar notation stands for the Clifford conjugation in $\G_m.$ Note that on the zero grade part of the complex Clifford algebra $\C_0=\C$ it coincides with the usual complex conjugation. The elements satisfying $\varphi^\dagger=\varphi$ and $\varphi^\dagger=-\varphi$ will be called Hermitian and anti-Hermitian respectively. Hermitian conjugation is a very important anti-involution which is the Clifford analogue of the conjugate transpose in matrices. It leads to the definition of the Hermitian inner product on $\C_m$ given by 
\begin{align} \label{HP}
\langle\varphi|\psi\rangle=[\varphi^\dagger \psi]_0, \;\; \varphi,\psi\in\C_m
\end{align}
where we recall that $[\;]_0$ denotes the projection to the scalar part, i.e. the grade zero part. Indeed, it is easy to see that it is linear in the second slot and conjugate linear in the first slot; for each $z\in\C$ and $\varphi\in\C_m$ we have
$
(z\varphi)^\dagger=\varphi^\dagger z^\dagger=\bar{z} \varphi^\dagger
$
since the Hermitian conjugation is anti-automorphism. The Hermitian symmetry of \eqref{HP} follows from the involutivity of  the Hermitian conjugation while its positive definiteness follows from the fact that
\begin{align*}
\langle\varphi|\varphi\rangle=[\varphi^\dagger \varphi]_0=\sum_{A}\varphi_A^2,
\end{align*}
where $A$ is an arbitrary multiindex and $\varphi_A$ is the coefficient at the Grassmann blade $e_A$, i.e. $\varphi=\sum_A\varphi_A e_A$. Let us discuss the last equality in more detail.
For two multi-indices $A=\{i_1,\dots,i_r\}$, $B=\{k_1,\dots,k_s\}$ the scalar projection 
$[\tilde{e}_Ae_B]_0$ is nonzero only if the grades are equal, i.e. $r=s$. Then by the definition of geometric product and the reverse operation we get $[\tilde{e}_Ae_B]_0=(e_{i_1}\cdot e_{k_1})\cdots(e_{i_r}\cdot e_{k_r})=\delta_{i_1k_1}\cdots \delta_{i_rk_r}=\delta_{AB}$, whence by the linearity of the grade projection
\begin{align*}
[\varphi^\dagger \varphi]_0=\sum_{A,B}\bar{\varphi}_A\varphi_B[\tilde{e}_Ae_B]_0=\sum_{A}\varphi_A^2.
\end{align*}

\subsection{Witt basis} \label{Witt}
Henceforth we assume the dimension of the generating vector space is even, i.e. $m=2n.$ In such a case the complexification of the Clifford algebra  can be introduced by considering so--called complex structure, i.e. a specific orthogonal linear transformation $J:\:\R^{2n}\to\R^{2n}$ such that $J^2=-1$, where 1 stands for the identity map. Namely, we choose $J$ such that its action on the orthonormal basis $e_1,\dots,e_{2n}$ is given by $J(e_j)=-e_{j+n}$ and $J(e_{j+n})=e_j$, $j=1,\dots,n.$ With $J$ one may associate two projection operators which produce the main objects of the complex setting by acting on the orthonormal basis, so–called Witt basis elements $(f_j , f^{\dagger}_j)$. Namely, we define
\begin{align*} 
f_j&=\frac12(1+iJ)(e_j)=\frac{1}{2}(e_j- ie_{j+n}), \;\; j=1,\dots,n\\
f_j^{\dagger} &=\frac{1}{2}(1- iJ)(e_j) =\frac{1}{2}(e_j+ ie_{j+n}),  \;\; j=1,\dots,n
\end{align*}
Note that it is not confusion of the notation since $f_j^\dagger$ indeed is the image of $f_j$ under  Hermitian conjugation \eqref{HC}. 
The Witt basis elements are isotropic with respect to the geometric product, i.e. for each $j=1,\dots,n$ they satisfy $f_j^2=0$ and $f_j^{\dagger}{}^2=0$. They also satisfy the Grassmann identities
\begin{align} \label{Grassmann}
f_j f_k + f_k f_j =f_j^{\dagger} f_k^{\dagger} + f_k^{\dagger} f_j^{\dagger}=0, \;\; j,k=1,\dots,n
\end{align}
and the duality identities
\begin{align} \label{duality}
f_j f_k^{\dagger} + f_k^{\dagger} f_j = \delta_{jk}, \;\; j,k=1,\dots,n
\end{align} 
The Witt basis of the whole complex Clifford algebra $\C_{2n}$ is then obtained, similarly
to the basis of the real Clifford algebra, by taking the $2^{2n}$ possible geometric products of
Witt basis vectors, i.e. it is formed by elements
\begin{align} \label{C2n_basis}
(f_1)^{i_1}(f_1^\dagger)^{j_1}\cdots (f_n)^{i_n}(f_n^\dagger)^{j_n}, \;\; i_k,j_k\in\{0,1\} \text{ for } k=1,\dots,n
\end{align}
One can eventually use the Grassmann blades of Witt elements as the basis of $\C_{2n}$. The relation of the two basis 
can be deduced from the relation of the geometric product of Witt basis elements to the corresponding inner and outer product, for more details see \cite{Lounesto}.
\begin{align*}
%f_j \cdot f_k &= f_j^{\dagger} \cdot  f_k^{\dagger} =0, \; 
%f_j \cdot f_k^{\dagger} = f_k^{\dagger} \cdot f_j = \frac{1}{2} \delta_{jk}, \;\; j,k=1,\dots,n \\
%f_j \wedge f_k &=- f_k \wedge f_j,\;
% f_j^\dagger \wedge f_k^\dagger =- f_k^\dagger \wedge f_j^\dagger, \;
% f_j^\dagger \wedge f_k =- f_k \wedge f_j^\dagger, \;\; j,k=1,\dots,n
f_jf_k&=f_j \cdot f_k + f_j \wedge f_k= f_j \wedge f_k \\
f_j^\dagger f_k^\dagger&=f_j^\dagger \cdot f_k^\dagger + f_j^\dagger \wedge f_k^\dagger= f_j^\dagger \wedge f_k^\dagger\\
f_j f_k^\dagger&=f_j \cdot f_k + f_j \wedge f_k^\dagger=\frac12 \delta_{jk}+ f_j \wedge f_k^\dagger
\end{align*}

\subsection{Spinor spaces}
In the language of Clifford algebras, spinor space is defined as a minimal left ideal of the complex Clifford algebra
and is realized explicitly by means of a self-adjoint primitive idempotent. The realization of spinor space within the complex Clifford algebra $\C_{2n}$ can be constructed directly using the Witt basis as follows. We start by defining
\begin{align*}
I_j=f_jf_j^\dagger \text{ and } K_j=f_j^\dagger f_j, \;\; j=1,\dots,n
\end{align*}
Direct computations show that  both $I_j, K_j$ are mutually commuting self--adjoint idempotents. More precisely, for $j,k=1,\dots,n$ the following identities hold.
\begin{align*}
I_j^\dagger&=I_j,\, I_j^2=I_j \text{ and } K_j^\dagger=K_j,\, K_j^2=K_j, \\
I_jI_k&=I_kI_j,\, K_jK_k=K_kK_j, \\
I_jK_k&=K_kI_j \text{ whenever } j\neq k, \text{ and } I_jK_j=K_jI_j=0.
\end{align*}
Moreover, the duality relations \eqref{duality} between Witt basis vectors imply that $I_j+K_j=1$ for each $j=1,\dots,n.$ Hence we get the resolution of the identity
$
1=\prod_{j=1}^n(I_j+K_j).
$
Consequently we get
\begin{align*}
\C_{2n}=\C_{2n}\prod_{j=1}^n(I_j+K_j)=\C_{2n}I_1\cdots I_n\oplus \C_{2n}I_1\cdots I_{n-1}K_1\oplus \cdots \oplus \C_{2n}K_1\cdots K_n,
\end{align*}
a direct sum decomposition of the complex Clifford algebra into $2^n$ isomorphic realizations of the spinor space that are denoted according to the specific idempotent involved:
\begin{align} \label{spinor}
\mathbb{S}_{\{i_1\cdots i_s\}\{k_1\cdots k_t\}}=\C_{2n}I_{i_1}\cdots I_{i_s}K_{k_1}\cdots K_{k_t}\subset\C_{2n},
\end{align}
where $s+t=n$ and the indices are pairwise different. %and $\{i_1\cdots i_s\}\cup\{k_1\cdots k_t\}=\{1,\dots,n\}$.
Each such space has dimension $2^n$ and its basis is obtained by right multiplication of the basis of $\C_{2n}$ by the corresponding primitive idempotent $I_{i_1}\cdots I_{i_s}K_{k_1}\cdots K_{k_t}.$ By the basic properties of the Witt basis elements \eqref{Grassmann} and \eqref{duality} it is easy to see that this action is nonzero if and only if the element of $\C_{2n}$ actually lies in the Grassmann algebra generated by $n$-dimensional space $( f_{i_1}^\dagger,\dots,f_{i_s}^\dagger,f_{k_1},\dots,f_{k_t} )$, i.e. we may write
\begin{align} \label{spinor_G}
\mathbb{S}_{\{i_1\cdots i_s\}\{k_1\cdots k_t\}}
=\Lambda( f_{i_1}^\dagger,\dots,f_{i_s}^\dagger,f_{k_1},\dots,f_{k_t} )I_{i_1}\cdots I_{i_s}K_{k_1}\cdots K_{k_t}.
\end{align}

In terms of multiindices $A=\{i_1,\dots,i_s\}$, $B=\{k_1,\dots,k_t\}$ this spinor space can be written shortly as $\Ss_{AB}$. 
%and its element can be written as $$\varphi=\sum_{C\subset A,D\subset B}\varphi_{CD} (f^\dagger)_C f_DI_{i_1}\cdots I_{i_s}K_{k_1}\cdots K_{k_t}$$.
It is an easy observation that it has the structure of  a Hilbert space of dimension $2^n$ due to the Hermitian product \eqref{HP}
%Indeed, for two spinors $\varphi,\psi\in\mathbb{S}_{AB}$ we have
%\begin{align*}
%\langle \varphi|\psi\rangle=\frac{1}{2^n}\sum_{C\subset A,D\subset B}\varphi_{CD}\psi_{CD}
%\end{align*}
and that the multiplication in $\C_{2n}$ makes each spinor space $\Ss_{AB}$ into a left $\C_{2n}$-module. Hence the elements of the complex Clifford algebra that keep the Hermitian product invariant define a representation of the corresponding unitary group on the spinor space. Namely let a $\lambda\in\C_{2n}$ act on two spinors $\varphi,\psi\in\Ss_{AB}$. Then we compute $\langle\lambda\varphi|\lambda\psi\rangle=[\varphi^\dagger\lambda^\dagger\lambda \psi]_0$ by definition and due to the antiautomorphism property of the Hermitian conjugation. Hence the elements of the complex Clifford algebra such that 
\begin{align} \label{unitary}
\lambda^\dagger \lambda=1
\end{align}
holds keep the Hermitian product invariant and thus define a representation of the unitary group $U(2^n)$ on the spinor space $\Ss_{AB}.$ These elements also satisfy $\lambda\lambda^\dagger=1$ and will be called unitary elements of $\C_{2n}$ in analogy with unitary matrices. Let us remark that in the representation theory this representation of the unitary group is well known. It comes from the so called spin representation of the corresponding complex orthogonal group.

\section{Quantum computing in complex Clifford algebras} \label{QCinC}
The idea is to perform quantum computing in the Hilbert space defined by a complex Clifford algebra with Hermitian product defined by \eqref{HP} instead of the classical realization of the Hilbert space on complex coordinate space with the standard Hermitian inner product. A quantum state is then represented by an element of a complex Clifford algebra lying in spinor space \eqref{spinor} and unitary transformations are then realized as elements \eqref{unitary} of the same algebra. The computation becomes especially efficient when using Witt basis of the complex Clifford algebra, see \ref{Witt}. However the mathematical framework described in the previous section allows for a direct application to general states and transformations of multiple qubits, for clarity we start with the description of the basic case of a single qubit.

\subsection{A qubit and single qubit gates} \label{qubit_section}
A qubit  will be represented by an element in the complex Clifford algebra $\C_2$ instead of its standard representation by a vector in the complex coordinate space $\C^2$. The Witt basis elements $(f,f^\dagger)$
satisfy the Grassmann and duality identities
\begin{align*}% \label{f}
f^2=f^\dagger{}^2=0, \, ff^\dagger+f^\dagger f=1,
\end{align*}
leading to $f f^\dagger =\frac12+ f\wedge f^\dagger$ and $[f f^\dagger]_0=\frac12$  in particular. The Witt basis vectors induce  a basis of the complex Clifford algebra $\C_2$ of the form $(1,f,f^\dagger,ff^\dagger)$. In this algebra we have two primitive idempotents $I=ff^\dagger$ and $K=f^\dagger f$ that give rise to two isomorphic spinor spaces: $\Ss=\C_{2}I=\Lambda(f^\dagger)I$ and $\bar{\Ss}=\C_{2}K=\Lambda(f)K.$ For representing the qubit we choose the former one, see Remark \ref{rem_notation}. Choosing  the basis $(1,f^\dagger)$ of the Grassmann algebra $\Lambda(f^\dagger)$, we get the following basis of $\Ss$ that will represent  the zero state and the one state of a qubit
\begin{align} \label{01}
\begin{split}
\ket{0}&=I=f f^\dagger \\ %\text{ and }  
\ket{1}&=f^\dagger I=f^\dagger f f^\dagger=(1-f f^\dagger) f^\dagger=f^\dagger,
\end{split}
\end{align}
whence the Clifford algebra representation of a qubit in a general superposition state  $\ket{\psi}=\alpha\ket{0}+\beta\ket{1}$, for arbitrary complex numbers $\alpha,\beta$,  is given by
\begin{align} \label{qubit}
\psi=(\alpha+\beta f^\dagger)I=\alpha f f^\dagger+\beta f^\dagger \in \Ss\subset\C_2.
\end{align}
However the basis \eqref{01} of spinor space $\mathbb{S}$ is orthogonal with respect to the Hermitian product in Clifford algebra $\C_2$ defined by \eqref{HP}, it is not orthonormal since the length of the basis elements equals $1/2$ due to the spinorial nature of the representation. To make the basis orthonormal we will modify the Hermitian product in \eqref{HP} by this factor, namely we will assume  $\langle\varphi|\psi\rangle=2[\varphi^\dagger \psi]_0$ for any spinors $ \varphi,\psi\in\C_2.$ Indeed, then we compute
\begin{align*}
\braket{0|1}&=2[f f^\dagger f^\dagger]_0=0,\\
\braket{0|0}&=\braket{1|1}=2[f f^\dagger]_0=1.
\end{align*}
\begin{remark} \label{rem_notation}
	The choice of idempotent $I$ is motivated by conventions in physics for creation and annihilation operators. Indeed, $f^\dagger$ is a realization of the abstract creation operator of the so-called CAR algebra and thus we want it to represent the qubit state $\ket{1}$ rather than $\ket{0}$.%in the sense of \eqref{01}. 
\end{remark}

A single qubit gate is represented by an unitary element in Clifford algebra $\C_2$ and it acts on a qubit in spinor space $\Ss$ by left multiplication. Obviously the identity gate is defined by $1\in\C_2$ and serially wired gates are given by the product of the individual representatives in $\C_2$ due to the associativity of the Clifford product. For our choice of the basis of qubit states the commonly used quantum gates operating on a single qubit are represented in terms of the Witt basis as follows.

\begin{proposition} \label{single_gates}
	Representing the basic qubit states in the complex Clifford algebra $\C_2$  as $\ket{0}=ff^\dagger$, $\ket{1}=f^\dagger$
	%For qubit $\psi=(\alpha+\beta f^\dagger)I$ 
	we get representations of single qubit gates in $\C_2$
	\begin{align*}
	\text{X-gate: } \: \lambda_X&=f^\dagger +f \\
	\text{Y-gate: } \: \lambda_Y&=if^{\dagger} -i f \\
	\text{Z-gate: } \:  \lambda_Z&=ff^{\dagger} -f^{\dagger} f %\\
%	\text{H-gate: } \:  \lambda_H&=\tfrac{1}{\sqrt{2}}(ff^{\dagger} -f^{\dagger}f + f +  f^{\dagger})
	\end{align*}
\end{proposition}
\begin{proof}
By the definition of Hermitian conjugation in Clifford algebra $\C_2$, the identification of basic qubit states  $\ket{0}=ff^\dagger$, $\ket{1}=f^\dagger$ leads to the identification of their Hermitian duals $\bra{0}= ff^\dagger$, $\bra{1}= f$. Then we get a representation of projection operators 
	\begin{align*}
	\ket{0}\bra{0}=ff^\dagger ff^\dagger=ff^\dagger, \;
	\ket{0}\bra{1}=ff^\dagger f=f, \;
	\ket{1}\bra{0}=f^\dagger ff^\dagger=f^\dagger, \;
	\ket{1}\bra{1}=ff^\dagger,
	\end{align*}
	where we used the Grassmann and duality identities for the Witt basis elements $f, f^\dagger.$
	The representations of single qubit gates from the proposition then  follow by their definitions on basic qubit states.
\end{proof}
\begin{remark}
	Using basis $1,e_1,e_2,e_1\wedge e_2$ of Clifford algebra $\C_2$ generated by an orthonormal basis $e_1, e_2$ of $\C^2$ instead of the Witt basis the representation of single qubit gates $X,Y$ and $Z$ are given as  
		\begin{align*}
 \lambda_X=e_1, \;
\lambda_Y=-e_2, \; 
  \lambda_Z=i e_{1} \wedge e_2 
	\end{align*}
\end{remark}

\begin{example}
Let us discuss the some of these basic quantum gates in more detail. The $X$-gate is the quantum equivalent to of the NOT gate for classical computers,  sometimes called a bit-flip as it maps the basis state $\ket{0}$ to $\ket{1}$ and vice versa. Hence we have
\begin{align*}
X=\ket{1}\bra{0}+\ket{0}\bra{1}=f^\dagger +f.
\end{align*}
which is equal to $e_1$ by definition of the Witt basis elements.
We can also compute directly the action of the corresponding element of $\C_2$ on a qubit basis in $\Ss$ to prove the correctness of the representation
\begin{align*}
\lambda_X \ket{0}&=(f^\dagger +f)f f^\dagger=f^\dagger f f^\dagger=f^\dagger =\ket{1}, \\
\lambda_X \ket{1}&=(f^\dagger +f) f^\dagger=f f^\dagger=\ket{0}.
\end{align*}
Similarly, for the  phase-flip $Z$-gate we get $Z=\ket{0}\bra{0}-\ket{1}\bra{1}= f f^\dagger - f^\dagger f$ since it leaves the basis state $\ket{0}$ unchanged and maps $\ket{1}$  to $-\ket{1}$   and a general phase-shift gate $\ket{1}\mapsto e^{i\varphi}\ket{1}$ is given by $R_\varphi=f f^\dagger + e^{i\varphi} f^\dagger f.$
The effect of a series circuit where $X$ is put after $Z$ can be described as a single gate represented by the Clifford product
\begin{align*}
XZ=(f^\dagger +f)(f f^\dagger - f^\dagger f)=f^\dagger-f.
\end{align*} 
It is also easy to check the involutivity of the single qubit gates $X,Y$ and $Z$, e.g. for a serial composition of two  $X$-gates we have
 \begin{align*}
 X^2=(f^\dagger +f)(f^\dagger +f)=f^\dagger f+f^\dagger f = 1.
 \end{align*} 
 The representations of rotation operator gates can be obtained directly by computing exponentials in of gates $X,Y,Z$ in $\C_2$. Consequently one can get the formula for the Hadamard gate, an important single qubit gate that we have not discussed yet. Namely, since $Y^2=1$, we compute
  \begin{align*}
 H&=X\operatorname{exp}(-iY\pi/4)=(f^\dagger+f)(\cos\tfrac{\pi}{4}+\sin\tfrac{\pi}{4}(f^\dagger-f))\\
 &=\tfrac{1}{\sqrt{2}}(ff^{\dagger} -f^{\dagger}f + f +  f^{\dagger}),
 \end{align*} 
which is equal to $\tfrac{1}{\sqrt{2}}(e_1+ie_1\wedge e_2)$ in the orthonormal basis.
\end{example}

The Clifford algebra representations of basic single qubit gates $X,Y$ and $Z$ in Proposition \ref{single_gates} determine an explicit form of a unitary element in $\C_2$ representing a general single qubit gate in terms of the Witt basis.
\begin{corollary} \label{1unitary_el}
	Each single qubit gate operating on qubit \eqref{qubit} is represented by an element of the complex Clifford algebra $\C_2$ of a form
	\begin{align} \label{1gate}
	\lambda=a f f^{\dagger} + b f +c f^{\dagger} + d f^{\dagger}f.
	\end{align}
	where $a,b,c,d \in\C$ are complex numbers such that $a^2 +  c^2=b^2  + d^2=1,\bar{b}a    +\bar{d}c  =0$ holds.
\end{corollary}
\begin{proof}
	Can be deduced from the description of projection operators in the proof of Proposition \ref{single_gates} or directly by writing the condition on unitary elements $\lambda^\dagger \lambda=1$ in terms of the Witt basis, see \eqref{unitary}, as follows. An arbitrary element $\lambda\in \C_2$ can be written as in \eqref{1gate} for some complex numbers $a,b,c,d\in\C$ since 	the four--tuple $(f f^\dagger, f^\dagger, f, f^\dagger f)$ form a basis of $\C_2$. The right-hand side of the equation for unitary elements can be written as  $ff^\dagger+f^\dagger f$ and for the left-hand side  we compute
	\begin{align*}
	\lambda^{\dagger} \lambda & =  (\bar{a} f f^{\dagger} + \bar{b} f^\dagger +\bar{c} f + \bar{d} f^{\dagger}f) (a f f^{\dagger} + b f +c f^{\dagger} + d f^{\dagger}f)\\
	&=  (a^2 +  c^2 )f^{\dagger}f  + 	( \bar{b}a    +\bar{d}c) f^{\dagger} +
	( \bar{a}b   +\bar{c}d) f    + ( b^2  + d^2 )f f^{\dagger},
	\end{align*}
	where we used  the definition of the Hermitian conjugation and its properties discussed in section \ref{ClC} and
	where we repeatedly used the duality and Grassmann identities for the Witt basis elements $f, f^\dagger$,
	% leading to identities
	%$
	%ff^\dagger 	ff^\dagger = (1-f^\dagger f)	ff^\dagger=	ff^\dagger, ff^\dagger f^\dagger=0 %, f f^\dagger f = f (1-f %f^\dagger) =f
	%$
	%etc., 
	see section \ref{Witt}. The result follows by comparing the coefficients on both sides of the equation while noting that $ \bar{a}b   +\bar{c}d$ is the complex conjugate of $ \bar{b}a    +\bar{d}c$.
\end{proof}

%Indeed, representations of the quantum gates in Proposition \ref{single_gates} are obtained by a suitable choice of complex  coefficients $a,b,c,d$ in equation \eqref{1gate}. 
Note that the condition on these coefficients can be expressed equivalently as the orthonormality of complex vectors $(a,c)$ and $(b,c)$ with respect to the standard Hermitian product on $\C^2.$ Hence the coefficients define a $2 \times 2$ unitary matrix proving  the equivalence between Clifford and matrix descriptions. Namely, each unitary element in $\C_2$ corresponds to a matrix in $U(2)$ as follows.
\begin{align} \label{correspondence}
a f f^{\dagger} + b f +c f^{\dagger} + d f^{\dagger}f 
\leftrightarrow
\begin{pmatrix}
a & b \\ c & d 
\end{pmatrix}%\in SU(2)
\end{align}
%The compatibility of the Clifford multiplication in $\C_2$ and matrix multiplication in $U(2)$ can be checked by computation of the action of the Clifford element in \eqref{correspondence} on spinor qubit $\psi=(\alpha+\beta f^\dagger)I$
%\begin{align*}
%\lambda \psi=(a f f^{\dagger} + b f +c f^{\dagger} + d f^{\dagger}f)(\alpha+\beta f^\dagger)I
%=
%(\alpha a + \beta b)I+ (\alpha c    + \beta d) f^{\dagger} I
%\end{align*}
%Indeed, the right-hand side is the spinor representation of a qubit whose vector representation is obtained by multiplication by the matrix in \eqref{correspondence}.
\begin{remark}
	Upon restriction to a normalized qubit $\langle\psi|\psi\rangle=2[\psi^\dagger \psi]_0=\alpha^2+\beta^2=1$, it is sufficient to consider gates from the special unitary group $SU(2)$, the connected component of $U(2).$ Such gates  are represented by unitary matrices with unite determinant and that they can be written as displayed in \eqref{correspondence} for $b=-\bar{c}$ and $d=\bar{a}$. Hence the elements in $\C_2$ representing subgroup $SU(2)$ are of a form $\lambda=a f f^{\dagger} -\bar{c} f +c f^{\dagger} + \bar{a} f^{\dagger}f.$
\end{remark}

\subsection{Multiple qubits and multiple qubit gates}
Following constructions in section \ref{CCA} the Hilbert space of states of a general $n$-qubit can be represented by a spinor space in the complex Clifford algebra $\C_{2n}$ and  $n$-qubit gates as unitary elements in the same algebra. For the explicit description we choose Witt basis $(f_1,f_1^\dagger,\dots,f_n,f_n^\dagger)$ of complex coordinate space $\C^{2n}$ leading to the Witt basis of the Clifford algebra $\C_{2n}$ formed by $2^{2n}$ geometric products of these elements given by \eqref{C2n_basis}. From the $2^n$ spinor spaces contained in the algebra  we choose the spinor space $\Ss_n=\C_{2n}I$ defined by primitive idempotent
\begin{align} \label{In}
I=I_{1}\cdots I_{n}=f_1 f_1^\dagger \cdots f_n f_n^\dagger
\end{align}
for modelling states of a $n$-qubit. This choice is motivated by its identification with the Grassmann algebra generated by "creation operators" $f_1^\dagger,\dots,f_n^\dagger$. Indeed, for such a realization of the spinor space of a $n$-qubit we have 
\begin{align} \label{Sn}
\Ss_n=\C_{2n}I=\Lambda(f_1^\dagger,\dots,f_n^\dagger)I
\end{align}
since  $f_jI=0$ for each $j=1,\dots,n$ by  Grassmann and duality identities for the Witt basis elements.
Similarly to the case of a single qubit, we multiply the Hermitian product \eqref{HP} by a normalization factor $2^n$ that reflects the spinorial nature of our representation in order to get a simple formula for elements of unite norm. Namely, for two spinors $\varphi,\psi\in\Ss_n$ we set
\begin{align} \label{HPn}
\langle\varphi|\psi\rangle=2^n[\varphi^\dagger \psi]_0.
\end{align}
With this choice of spinor space and Hermitian product the main results of section \ref{CCA} that we need for representing  qubits and quantum gates in a complex Clifford algebra read as follows.
\begin{proposition} \label{n_qubit}
	Spinor space $\Ss_n\subset\C_{2n}$ given by \eqref{Sn} together with Hermitian product \eqref{HPn} form a Hilbert space of dimension $N=2^n$ with an orthonormal basis 
	\begin{align} \label{Sn_basis}
	\ket{i_1 \cdots i_n}=(f_1^\dagger)^{i_1} \cdots (f_n^\dagger)^{i_n} I,
	\end{align}
	where $ i_1,\dots,i_n\in\{0,1\}$. Unitary transformations are given by left multiplications by unitary elements, i.e. elements of $\C_{2n}$ such that $\lambda^\dagger \lambda=1$.
\end{proposition}
\begin{proof}
	The proposition follows from the constructions described in section  \ref{CCA}, we only need to check the orthonormality of basis elements \eqref{Sn_basis}. For $ i_k,j_k\in\{0,1\}$, where $k=1,\dots,n$, the Hermitian product of two basis elements is given by
	\begin{align*}
	\braket{j_1 \cdots j_n|i_1 \cdots i_n}=2^n[I^\dagger(f_n)^{j_n} \cdots (f_1)^{j_1} (f_1^\dagger)^{i_1} \cdots (f_n^\dagger)^{i_n} I]_0
	\end{align*}
	by definition. Let us prove the orthogonality first. Assume  $i_k=0$ and $j_k=1$ for some $k$. Since the Witt basis element $f_k$ anti-commutes with elements $f_\ell$  and $f_\ell^\dagger$  for each $\ell\neq k$ by \eqref{Grassmann},  the previous formula can be expressed in a form $\pm 2^n[I^\dagger  \cdots  f_kI]_0$  and it vanishes since $f_kI=0$ by \eqref{Grassmann} and \eqref{duality}. Similarly, if  $i_k=1$ and $j_k=0$, then the above formula for Hermitian product vanishes since it contains factor $I^\dagger f_k^\dagger=0.$ Hence the Hermitian product vanish if $i_k\neq j_k$ for some $k$ and the the orthogonality is proven. To prove the normality we notice that $ (f_k)^{i_k} (f_k^\dagger)^{i_k}=I_{i_k}$ is an idempotent commuting with all elements $f_\ell, f_\ell^\dagger$, where $\ell\neq k,$ and all idempotents $I_{i_\ell}$. Hence we get
		\begin{align*}
		\braket{i_1 \cdots i_n|i_1 \cdots i_n}=2^n[I^\dagger(f_n)^{i_n} \cdots (f_1)^{i_1} (f_1^\dagger)^{i_1} \cdots (f_n^\dagger)^{i_n} I]_0=2^n[I^\dagger \cdots I_{i_\ell} \cdots I]_0,
	\end{align*}
	where $I_{i_\ell}$ are idempotents  such that $i_\ell=1.$ The primitive idempotent $I$ satisfies $I^\dagger=I$ and it also satisfies  $I_{i_\ell}I=I$ for all $\ell$ since it is given by product of all such commuting idempotents, namely  $I=I_1\cdots I_n$ by definition. So we compute
	\begin{align*}
		\braket{i_1 \cdots i_n|i_1 \cdots i_n}=2^n[ I]_0=2^n[I_1\cdots I_n]_0=1,
	\end{align*}
	where the last equality follows from the decomposition of idempotents into grade components. Namely, we have
	$I_k=1/2+f_k\wedge f_k^\dagger$ and the geometric product of $f_k\wedge f_k^\dagger$ with elements not containing $f_k$ neither $f_k^\dagger$ either vanishes or yields an element of grade at least two.
\end{proof}

\begin{remark}
Note that we chose MSB bit numbering for $n$-qubits. The choice of the LSB bit numbering would lead to different but isomorphic representations in $\C_{2n}.$  
\end{remark}

However an explicit description of a general unitary element in $\C_{2n}$ representing a $n$-qubit gate similar to the description of a general single qubit gate given in Corollary \ref{1unitary_el} is possible, it is more sophisticated and thus not helpful. The same happens in matrix representation and it reflects the complexity of unitary group $U(N)$.
On the other hand, given a  specific $n$-qubit gate its representation in $\C_{2n}$ is obtained by rewriting the defining formula in terms of projection operators in Dirac formalism via identification of $n$-qubit states \eqref{Sn_basis}. To make it clear we elaborate some examples for $n=2$ in more detail. 
\begin{example} \label{ex_2gates}
	In the case of a 2-qubit we work in Clifford algebra $\C_{4}$ of dimension is $2^4=16$. Using the Witt basis $f_1,f_1^\dagger,f_2,f_2^\dagger$ of $\C^4$ we define a primitive idempotent $I=f_1 f_1^\dagger f_2 f_2^\dagger$ which gives rise to spinor space $\Ss_2=\C_4I$ of dimension $2^2=4$ with an orthonormal basis
 \begin{align*}
 \ket{00}=I,\,  \ket{10}=f_1^\dagger I,\, \ket{01}=f_2^\dagger I,\,  \ket{11}=f_1^\dagger f_2^\dagger I
 \end{align*}
 Using this representation of basis states and the definition of Hermitian conjugation we can form 16 projection operators, e.g. $\ket{00}\bra{00}=I I^\dagger=I$, $\ket{00}\bra{01}=I I^\dagger f_2 = I f_2=f_1 f_1^\dagger f_2 ,$ etc. Specific 2-qubit gates are then formed by a complex linear combinations of these elements in $\C_4.$  We demonstrate the functionality of the spinor representation on 2-qubit gates known as  CNOT, CZ and SWAP, see the diagrammatic descriptions of these gates in Figure \ref{2gates}.
  \begin{align*}
 f_1f_1^\dagger \ket{00}&=f_1f_1^\dagger f_1f_1^\dagger f_2f_2^\dagger=I = \ket{00},\\
 f_1f_1^\dagger \ket{01}&=f_1f_1^\dagger f_2^\dagger f_1f_1^\dagger f_1f_1^\dagger =  f_2^\dagger I =\ket{01}
 \end{align*}
\begin{figure}[h] 
	\centering
	\includegraphics[width=2cm,valign=t]{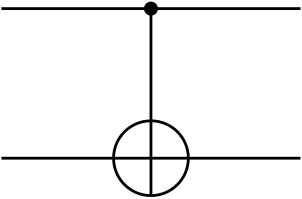}
	\hspace{2em}
    \includegraphics[width=2cm,valign=t]{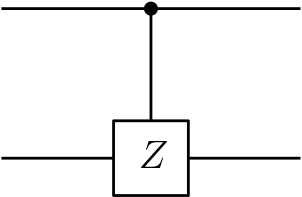}
	\hspace{2em}
   \includegraphics[width=2cm,valign=t]{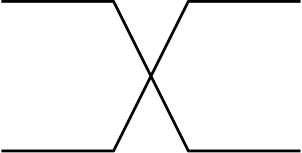}
	\caption{2-qubit gates CNOT, CZ and SWAP respectively}
	\label{2gates}
\end{figure}
 We systematically use Grassmann and duality identities for Witt basis elements, identity $II^\dagger=I$ in particular.
 \begin{align*}
 \lambda_{\operatorname{CNOT}}&=\ket{00}\bra{00}+\ket{01}\bra{01}+\ket{11}\bra{10}+\ket{10}\bra{11}\\
 &=I+f_2^\dagger I f_2+f_1^\dagger f_2^\dagger I f_1+ f_1^\dagger I f_2 f_1\\
 &=f_1 f_1^\dagger f_2 f_2^\dagger +f_1 f_1^\dagger f_2^\dagger f_2 - f_1^\dagger f_1 f_2^\dagger - f_1^\dagger f_1 f_2 \\
 &=f_1f_1^\dagger -  f_1^\dagger f_1 (f_2^\dagger + f_2), \\
   \lambda_{\operatorname{CZ}}&= \ket{00}\bra{00}+\ket{01}\bra{01}+\ket{10}\bra{10}-\ket{11}\bra{11}\\
 &=f_1f_1^\dagger + f_1^\dagger f_1 (f_2f_2^{\dagger} -f_2^{\dagger} f_2), \\
    \lambda_{\operatorname{SWAP}}&=\ket{00}\bra{00}+\ket{11}\bra{11}+\ket{10}\bra{01}+\ket{01}\bra{11}\\
 %  &=I I^\dagger+f_1^\dagger f_2^\dagger I I^\dagger f_2 f_1+f_1^\dagger I I^\dagger f_2 + f_2^\dagger I I^\dagger f_1\\
 &=f_1f_1^{\dagger}f_2f_2^{\dagger}+
 f_1^{\dagger}f_1f_2^{\dagger}f_2+f_1^{\dagger}f_2-f_1f_2^{\dagger}.
 \end{align*}
\end{example}

\subsection{Tensor product of gates}
%We have described the basic quantum gates.
To describe effectively quantum logic circuits in the complex Clifford algebra, it remains to discuss representations of parallel quantum gates, i.e. representations of tensor product of gates. First of all we realize that the representation of tensor products of states is already determined by  Proposition \ref{n_qubit}. Namely, a $n$-qubit  $\ket{i_1\cdots i_n}=\ket{i_1}\otimes\cdots\otimes\ket{i_n}$ is in the Clifford algebra represented by geometric product of representations of individual qubits. Hence a tensor product $\ket{\varphi} \otimes \ket{\psi}$ is represented by geometric product $\varphi\psi$, where the spinors $\varphi,\psi$ are assumed to lie in disjoint vector spaces viewed as two orthogonal subspaces of their union. Now consider an action of a tensor product of gates $\lambda\otimes \mu$ given  by unitary elements $\lambda,\mu$ of the Clifford algebra on such a state. The resulting state $\lambda \varphi \otimes \mu \psi$  is represented by $\lambda \varphi\mu \psi$ which is different from $\lambda \mu \varphi\psi$ in general due to the skew-symmetry of the geometric product. Namely, for two blades $e_A,e_B$ determined by disjoint multi-indices $A,B$ we have 
\begin{align*}
e_A e_B = (-1)^{|A| |B|} e_B e_A.
\end{align*}
and so the Clifford algebra has the structure of a superalgebra. Hence the geometric product does not represent the ordinary tensor product but it represents the super tensor product. It has the same structure as a vector space but with the multiplication rule determined by
\begin{align} \label{STP}
(e_A e_B) (e_C  e_D)=(-1)^{|B| |C|}(e_Ae_C )( e_Be_D)
\end{align}
on blades.
Consequently, the geometric product identifies complex Clifford algebra for $n$-gates  with the super tensor product of Clifford algebras for single qubit gates. The ordinary ungraded tensor product $\lambda_1 \otimes \cdots \otimes \lambda_n$ of gates $\lambda_k$ from distinct copies of $\C_2$ is in Clifford algebra $\C_{2n}$ represented by $ \lambda_1\cdots \lambda_n$ only up to the sign. Although this sign depends on the $n$-qubit on which we act by \eqref{STP}  in general, it is completely determined by the set $\{\lambda_1,\dots,\lambda_n\}$ in the case that $\lambda_k$ for each $k=1,\dots,n$ is one of the basis elements of $\C_2$
\begin{align} \label{lambda_k}
\lambda_k\in\{f_k f_k^\dagger,f_k^\dagger f_k, f_k, f_k^\dagger\}.
\end{align}

\begin{proposition} \label{TP}
A tensor product $ \lambda_1 \otimes \cdots \otimes \lambda_n$, where  $\lambda_k \in \{ f_k f_k^{\dagger},f_k^{\dagger} f_k, f_k, f_k^{\dagger}   \} $ for each $ k = 1, \dots , n$, is represented by geometric product $(-1)^s \lambda_1 \cdots \lambda_n$,  where the sign is determined by the cardinality of the sets $S_i$, such that  $s=\sum_i |S_i|$, where 
\begin{align} \label{S}
    S_i=  
    \{ \ell < i \: : \: \lambda_l =f_{\ell} \text{ or }  \lambda_{\ell} = f_{\ell}^{\dagger}f_{\ell} \}
    \text{ in the case if  }{\lambda_i = f_i \text{ or }\lambda_i
    =f_i^{\dagger}} .
\end{align}
%	Tensor product  $\lambda_1 \otimes \cdots \otimes \lambda_n$, where $\lambda_k$ lies in the set \eqref{lambda_k} for each $k=1,\dots,n$, is represented by geometric product   $(-1)^{|S|} \lambda_1\cdots \lambda_n$, where the sign is given by the cardinality of the set
%\begin{align}
%S=\bigcup\limits_{i\::\: \lambda_i=f_i \text{ or }\lambda_i=f_i^\dagger}\{\ell <i \;:\; \lambda_\ell=f_\ell \text{ or } \lambda_\ell=f_\ell^\dagger f_\ell \}
%\end{align}
\end{proposition}
\begin{proof}
A  $n$-qubit $\psi_1 \otimes \cdots \otimes \psi_n$ is represented by geometric product $\psi_1\cdots\psi_n$ of mutually orthogonal spinors $\psi_k$ by \eqref{Sn_basis}.
Representation of a qubit obtained upon the action of $\lambda_1 \otimes \cdots \otimes \lambda_n$ on this $n$-qubit  is given by  $$\lambda_1\psi_1 \cdots \lambda_n\psi_n=(-1)^p\lambda_1\cdots\lambda_n\psi_1\cdots\psi_n$$ since $\lambda_k,\psi_k$ are orthogonal to $\lambda_\ell,\psi_\ell$ for $k\neq \ell.$ Roughly speaking, the sign is determined by how many times we need to commute to get all elements $\lambda_k$ to the left hand side. Thus it depends on spinors $\psi_\ell$, $\ell<k$ which are combinations of components of grade one or two in general. Only the commuting with grade one components does change the sign. However the grade one spinors are multiples of $f_\ell^\dagger$ and they are annihilated by all elements of $\C_{2n}$ except  elements $\lambda_\ell=f_\ell$ and  $\lambda_\ell=f_\ell^\dagger f_\ell$ which act nontrivially. Hence the number of commutation steps that push $\lambda_k$ to the left is equal to the number of such elements $\lambda_\ell,$ $\ell<k.$
\end{proof}

\begin{example} Let us construct 2-qubit gates 	$X \otimes Y$ and $Y \otimes X$ according to Proposition \ref{TP}. First we write these gates as a sum of tensor products of basis Witt basis elements and then for each such summand we compute the cardinality of set $S$ giving the sign of the corresponding geometric product.
	\begin{align*}
	X \otimes Y &= i(f_1^\dagger \otimes f_2^\dagger-f_1^\dagger \otimes f_2+f_1 \otimes f_2^\dagger-f_1 \otimes f_2)\\
	&=i(f_1^\dagger f_2^\dagger-f_1^\dagger  f_2-f_1  f_2^\dagger+f_1  f_2)\\
	Y \otimes X &= i(f_1^\dagger \otimes f_2^\dagger+f_1^\dagger \otimes f_2-f_1 \otimes f_2^\dagger-f_1 \otimes f_2)\\
	&=i(f_1^\dagger f_2^\dagger+f_1^\dagger  f_2+f_1  f_2^\dagger+f_1  f_2)
	\end{align*}
	Let us assume even more simple example of a $X$-gate with a parallel qubit without any gate. If the gate is acting on the first qubit we get a resulting 2-qubit gate $X \otimes \operatorname{id}=X_1$. However, acting on the second qubit we need to write the identity representation as $1=f_1 f_1^\dagger+f_1^\dagger f_1$ since idempotent $K_1=f_1^\dagger f_1$ makes the change of sign in contrast to idempotent $I_1=f_1 f_1^\dagger$,
	\begin{align*}
	\operatorname{id} \otimes X=(I_1+K_1) \otimes X_2=I_1X_2-K_1X_2=f_1 f_1^\dagger(f_2^\dagger +f_2)-f_1^\dagger f_1(f_2^\dagger +f_2).
	\end{align*}
	The representations of controlled gates from example \ref{ex_2gates} can be constructed from tensor product of single qubit gates as follows.
	\begin{align*}
\lambda_{\operatorname{CNOT}}&=I_1\otimes 1 +K_1 \otimes X_2=I_1-K_1 X_2\\
\lambda_{\operatorname{CZ}}&=I_1\otimes 1 +K_1 \otimes Z_2=I_1+K_1 Z_2
	\end{align*}

\end{example}

\begin{example}
	The spinor space $\Ss_3$ representing states of 3-qubits has dimension $2^3=8$ in Clifford algebra $\C_{6}$ of dimension $2^6=64$. Using the primitive idempotent \eqref{In} and the orthonormal basis representation \eqref{Sn_basis} in terms of the Witt basis the Toffoli gate is represented by
	\begin{align*}
	\lambda_{\operatorname{CCNOT}}&=(I_1\otimes I_2  +I_1\otimes K_2 +I_2\otimes K_1)\otimes \operatorname{id} +K_1\otimes K_2\otimes X_3\\
	&=I_1I_2+I_1K_2+I_2K_1+K_1K_2X_3=1-K_1K_2+K_1K_2X_3\\
	&=1+f_1^\dagger f_1f_2^\dagger f_2(f_3 + f_3^\dagger-1). \\
\lambda_{\operatorname{CSWAP}}&=I_1\otimes \operatorname{id}+K_1\otimes \lambda_{\operatorname{SWAP}} \\
&=I_1+K_1(I_2I_3+K_2K_3+f_2^{\dagger}f_3-f_2f_3^{\dagger})\\
&=f_1 f_1^\dagger + f_1^\dagger f_1(f_2 f_2^\dagger f_3 f_3^\dagger+f_2^\dagger f_2 f_3^\dagger f_3+f_2^{\dagger}f_3-f_2f_3^{\dagger})
%&=1+K_1K_2I_3(f_2^\dagger f_3-1)+K_1I_2K_3(f_2 f_3^\dagger-1)
\end{align*}
\end{example}

%\subsection{Grover algorithm}

\section{Quantum computing in real Clifford algebras}\label{realGA}	
Accidental isomorphism can be used to formulate intrinsically complex quantum computing in a real framework. We show two ways how to see a qubit in real Clifford algebra $\G_3$, i.e the GA induced by the standard euclidean inner product of signature $(3,0).$ The first approach appears in literature, see  \cite{Lasenby1, dl,cm}, and describes qubit states as even elements in this algebra or equivalently as unite quaternions. The second approach is new and follows from the complex representation of qubits described above. For the other know concepts see \cite{havel,soma,grego}. We also mention how to deal with multiple qubits and multiple qubit gates in the real case.

\subsection{A quaternionic qubit} \label{H_qubit}
The transition from complex  to real framework which appears in literature is based on the well known coincidental isomorphism of Lie algebras $\mathfrak{su}(2)\cong \mathfrak{so}(3)$, or more precisely, on the corresponding isomorphism of Lie groups
\begin{align} \label{SUSpin}
SU(2)\cong Spin(3),
\end{align}
and the isomorphism of these groups with the group of unite quaternions. We can easily describe these isomorphisms explicitly by realizing  Lie algebra $\mathfrak{so}(3)$ as bivectors in Clifford algebra $\G_3$ and Lie group $Spin(3)$ as elements of even grade in $\G_3$. Namely,  in terms of Pauli matrices the Lie algebra isomorphism can be defined by mapping $i_{\mathbb{C}}\sigma_k\mapsto i\sigma_k,$ $k=1,2,3$, where we denote  the usual complex unite by $i_{\mathbb{C}}$ in order to distinguish from pseudoscalar  $i=\sigma_{1}\sigma_2\sigma_3$ in $\G_3$, while $\sigma_k$ on the right hand side is seen as a vector in $\G_3$ satisfying $\sigma_k^2=1$. Consequently, using the Einstein summation convention, we get a Lie group isomorphism \eqref{SUSpin} of a form 
\begin{align} \label{iso}
\begin{pmatrix}
a^0+a^3 i_{\mathbb{C}} & a^2 +a^1 i_{\mathbb{C}} \\
-a^2+a^1 i_{\mathbb{C}} & a^0-a^3 i_{\mathbb{C}}
\end{pmatrix}
\mapsto 
a^0+a^1 \sigma_{2}\sigma_3+ a^2\sigma_{3}\sigma_1+a^3\sigma_{1}\sigma_2=a^0+a^k\sigma_k^*
\end{align}
where the coefficients $a^0,a^1,a^2,a^3\in\mathbb{R}$ and $\sigma_k^*=\sigma_k i = i \sigma_k$ is the duality defined by pseudoscalar $i=\sigma_{1}\sigma_2\sigma_3$. 
Assigning the quaternionic unites to $\sigma_k$, $k=1,2,3$, defines an isomorphism with unite quaternions.
A general state of a qubit is identified with the first column of the matrix on left hand side, thus in the real Clifford algebra $\G_3$ is represented by 
\begin{align*} 
\ket{\psi}=\begin{pmatrix}
a^0+a^3 i_{\mathbb{C}}  \\
-a^2+a^1 i_{\mathbb{C}} 
\end{pmatrix}
\leftrightarrow \psi = a^0 +a^k \sigma_k^*
\end{align*}
In particular, the standard computational basis $(1,0)$ and $(0,1)$ in $\C^2$ is in the real Clifford algebra formulation represented by 
\begin{align} \label{basisH}
\ket{0}=
%\begin{pmatrix}
%1 \\
%0
%\end{pmatrix}
%\leftrightarrow
1 \text{ and }
\ket{1}=
%\begin{pmatrix}
%0 \\
%1
%\end{pmatrix}
%\leftrightarrow
-i\sigma_2=\sigma_{1}\sigma_3,
\end{align}
respectively.
The identification \eqref{iso} also determines explicit formulas for a Hermitian inner product and a representation of Pauli matrices on even subalgebra $\G_3^0,$ namely for $\varphi,\psi\in\G_3^0$ we have  
\begin{align} \label{productH}
\braket{\varphi|\psi}&=[\tilde{\varphi}\psi]_0-[\tilde{\varphi}\psi\sigma_1\sigma_2]_0 i_{\mathbb{C}}, \\
 \label{PauliH}
\hat{\sigma}_k\ket{\psi} &\leftrightarrow \sigma_k\psi\sigma_3.
\end{align}
These formulas can be explained by viewing the unitary group $SU(2)$ as $SO(4)\cap GL(2,\C)$, i.e. as the group of orthogonal transformations with respect to a real scalar product of signature $(4,0)$ commuting with an orthogonal complex structure. A choice of a scalar product and a complex structure on even elements $\G_3^0$ then defines an Hermitian inner product on this space by a standard construction and thus defines an isomorphism \eqref{SUSpin}. In our case, the scalar product is given  by $(\varphi,\psi)=[\tilde{\varphi}\psi]_0$ and the complex structure $J$  is defined by $J\psi=\psi i\sigma_3=\psi \sigma_{1}\sigma_2$. Indeed, for such a choice  the Hermitian product \eqref{productH} is constructed as
\begin{align*} 
\braket{\varphi|\psi}=(\varphi,\psi)-(\varphi,J\psi)i_{\mathbb{C}}.
\end{align*}
The action of Pauli matrices in $\G_3^0$  given by \eqref{PauliH} keep the scalar product invariant and commutes with the complex structure and thus keeps this Hermitian product invariant.
% Indeed, it is easy to see that  the action of 
%\begin{align*} 
%\sigma_k\psi\sigma_3=-\sigma_k^*\psi\sigma_{12} = - J\sigma_k^*\psi.
%\end{align*}
\begin{remark}
This point of view also allows to see the freedom of quaternionic representation of qubits. Namely, choosing a different complex structure or modifying the scalar product on $\G_3^0$ would lead to an isomorphism \eqref{SUSpin} different from \eqref{iso} leading to representations of computational basis, Hermitian product and Pauli matrices different from \eqref{basisH}, \eqref{productH} and \eqref{PauliH}.
\end{remark}

The reality of this qubit representation implies that multiple qubits are represented in a quotient space defined by so called correlator. Namely, representing qubits in the  real geometric algebra $\G_3^+$  the space of $n$-qubits is $\G_3^+\otimes\cdots\otimes\G_3^+$ instead the  tensor power of $n$ copies of $\C^2$. However this  is the complex tensor product according to axions of the quantum mechanics. If we want to have a fully real description, including the real tensor product, we need to identify complex structures $J_k=i\sigma_3^k=\sigma_1^k\sigma_2^k$ (representing the multiplication by complex unite) in all copies. This can be done by introducing the $n$-qubit correlator
\begin{align*}
E_n=\prod_{k=2}^n \frac12 (1-i\sigma_3^1 i\sigma_3^k).
\end{align*} 
Indeed, this element satisfies $E_nJ_k=E_nJ_\ell$ for all $k,\ell = 1,\dots, n$ and thus it defines a quotient space $\G_3^+\otimes\cdots\otimes\G_3^+/E_n$ with a complex structure $J_n=E_nJ_k=E_ni\sigma_3^k$. Multivectors belonging to this space can be regarded as $n$-qubit states. 

\subsection{A real complex qubit}
Another way how to describe states of a qubit by a real algebra is to transfer its complex representation described in section \ref{qubit_section} via the accidental isomorphism of real algebras
\begin{align} \label{isoC2G3}
\C_2 \cong \G_3.
\end{align}
In order to obtain an explicit representation of a qubit we choose a concrete realization of  this isomorphism. Namely, 
in terms of the Witt basis of $\C_2$ and an orthonormal basis $\sigma_k$ of $\R^3$ we consider the isomorphism given by mapping
\begin{align*}
1&\mapsto 1  & i_{\mathbb{C}}&\mapsto \sigma_1\sigma_2\sigma_3\\
f&\mapsto \tfrac12(\sigma_1-\sigma_1\sigma_3) & i_{\mathbb{C}}f&\mapsto \tfrac12(\sigma_2\sigma_3-\sigma_2)\\
f^\dagger&\mapsto \tfrac12(\sigma_1+\sigma_1\sigma_3) & i_{\mathbb{C}}f^\dagger&\mapsto \tfrac12(\sigma_2\sigma_3+\sigma_2)\\
ff^\dagger&\mapsto \tfrac12(1+\sigma_3) & i_{\mathbb{C}}ff^\dagger&\mapsto \tfrac12(\sigma_1\sigma_2+\sigma_1\sigma_2\sigma_3)
\end{align*}
In particular, the complex unite $i_{\mathbb{C}}$ is mapped to trivector $\sigma_1\sigma_2\sigma_3\in\G_3$ and the primitive idempotent $I=f f^\dagger\in \C_2$ is mapped to real idempotent
\begin{align*}
I_\R= \frac12(1+\sigma_3) \in \G_3.
\end{align*}
Using this idempotent the equivalence between the classical description of a qubit as a complex vector and as an element of $\G_3$ based on this isomorphism reads
\begin{align*} 
\ket{\psi}=\begin{pmatrix}
a^0+a^3 i_{\mathbb{C}}  \\
a^1+a^2 i_{\mathbb{C}} 
\end{pmatrix}
\leftrightarrow \psi = (a^0 +a^1\sigma_1 + a^2\sigma_2 + a^3\sigma_1\sigma_2)I_\R,
\end{align*}
where $a^0,a^1,a^2,a^3\in\R.$
Note that, in contrast to the quaternionic representation described in the previous section, a qubit is represented by a multivector in $\G_3$ containing blades of both even and odd grades in this case. In particular, the computational basis is given by
\begin{align}
\ket{0}= I_\R\text{ and } \ket{1}=\sigma_1 I_\R.
\end{align}
The Hermitian inner product on $\G_3$ is given by the transition of the Hermitian product on $\C_2$ given by \eqref{HP} via isomorphism \eqref{isoC2G3}. Looking at the prescription of the isomorphism we see that the Hermitian conjugation of basis elements is mapped to the reverse of the corresponding images in $\G_3.$ Hence we have a particularly simple formula for the Hermitian product in this case, namely for two qubits $\varphi,\psi\in\G_3$ we have
\begin{align} 
\braket{\varphi|\psi}=[\tilde{\varphi}\psi]_0.
\end{align}
Our formula for isomorphism \eqref{isoC2G3} yields also a particularly simple formula for the representation of Pauli matrices, namely
\begin{align}
\hat{\sigma}_k\ket{\psi} &\leftrightarrow \sigma_k\psi.
\end{align}

Although this real representation of a qubit is quite elegant, the representation of multiple qubits is as complicated as in the case of the quaternionic qubit  described in \ref{H_qubit}. Due to its reality we need to use a correlator to identify the multiplication by complex unite in each slot of the tensor product $\G_3\otimes\cdots\otimes\G_3$. Since this is a common feature of all real description of qubits we believe that the right way is to use the complex GA as described in section \ref{QCinC} above.

\end{document}